\documentclass[submission,copyright,creativecommons]{eptcs}
\providecommand{\event}{AFL 2026} 

\usepackage{iftex}
\usepackage{amsmath}
\usepackage{amsthm}
\usepackage{bm}
\usepackage{lineno}

\ifpdf
\usepackage{underscore}         
\usepackage[T1]{fontenc}        
\else
\usepackage{breakurl}           
\fi

\usepackage{tikz}
\usetikzlibrary{intersections}
\usepackage{xspace}
\usepackage{cleveref}
\usepackage{thmtools,thm-restate}
\usepackage{multirow}
\usepackage{multicol}
\usepackage{tabularx,booktabs}
\newcolumntype{Y}{>{\centering\arraybackslash}X}
\usepackage{tikz-cd}
\usepackage{amsfonts}
\usepackage{caption}
\usepackage{subcaption}
\usepackage{newtxtext}
\usepackage{float}

 \newcommand{\cut}[1]{}

\newtheorem{theorem}{Theorem}
\newtheorem{definition}[theorem]{Definition}
\newtheorem{example}[theorem]{Example}
\newtheorem{remark}[theorem]{Remark}

\newtheorem{lemma}[theorem]{Lemma}
\newtheorem{proposition}[theorem]{Proposition}
\newtheorem{corollary}[theorem]{Corollary}

\def\squareforqed{$\Box$}
\def\QED{\ifmmode\squareforqed\else{\unskip\nobreak\hfil
\penalty50\hskip1em\null\nobreak\hfil\squareforqed
\parfillskip = 0pt\finalhyphendemerits = 0\endgraf}\fi}

\newcommand{\ignore}[1]{ }

\newcommand{\id}{\bm{id}}

\newcommand{\bA}{\bm{A}}
\newcommand{\bB}{\bm{B}}
\newcommand{\bC}{\bm{C}}

\newcommand{\bc}{\bm{c}}

\newcommand{\bsigma}{\bm{\sigma}}

\newcommand{\btau}{\bm{\tau}}
\newcommand{\bupsilon}{\bm{\upsilon}}

\newcommand{\btheta}{\bm{\theta}}
\newcommand{\bpsi}{\bm{\psi}}
\newcommand{\bomega}{\bm{\omega}}

\newcommand{\bchi}{\bm{\chi}}
\newcommand{\bvarphi}{\bm{\varphi}}

\newcommand{\FMA}{\text{\sf FMA}\xspace}

\newcommand{\HRA}{\text{\sf HRA}\xspace}

\newcommand{\tuple}[1]{( #1 )}

\newcommand{\size}[1]{|#1|}

\newcommand{\conf}[2]{[#1,#2]            }
\newcommand{\trans}[1]{ \xrightarrow{#1}  }
\newcommand{\transword}[1]{ \xrightarrow{#1}\mathrel{\vphantom{\to}^*}}
\newcommand{\langof}[1]{L(#1)}

\newcommand{\Reg}{\mathcal{R}}

\title{Pumping Constants for Infinite Alphabets}
\author{Yoav Danieli
\institute{Technion\\
Haifa, Israel}
\email{yoavd@campus.technion.ac.il}}
\def\titlerunning{Pumping Lemmas}
\def\authorrunning{Y. Danieli}
\begin{document}
\maketitle

\begin{abstract}
It is well known that
the pumping lemma for regular languages over finite alphabets
does not extend to infinite alphabets.
Instead, the generalized version of the pumping lemma for infinite alphabets
states that the pumped patterns need not be identical,
but only equivalent up to a finite-order permutation
of the alphabet.

This generalized variant of the pumping lemma involves two parameters:
the length of the pumped pattern and
an upper bound on the order of the permutation.
We improve the known permutation-order
bounds from factorial to tight bounds
governed by Landau's function.
We also show that, if the pumping length
is sufficiently large, then permutations
of order at most two always suffice.

Our main structural result concerns the one-register case:
every one-register automaton satisfies the classical
pumping lemma, with no alphabet permutation.
We prove a quadratic upper bound on the pumping length and give
a matching quadratic lower bound up to constant factors.
The same classical pumping phenomenon is shown
for hierarchical register automata.
Finally, we prove that, in general,
deciding whether a quasi-regular language
is pumpable for given constants is undecidable,
while the minimal classical pumping length
is computable for non-guessing one-register automata.
\end{abstract}

\section{Introduction}

{\em Finite-memory automata}~\cite{KaminskiF94}
are a generalization of
the classical Rabin-Scott finite-state automata~\cite{RabinS59} to
{\em infinite} alphabets.
By equipping the automata with a finite set
of registers (called {\em windows} in~\cite{KaminskiF94}),
which store letters from the infinite alphabet
during the computation, and
by restricting the automaton to
compare input letters with register contents
and to copy input letters to registers only,
the automaton retains a finite memory of previously seen letters.
Consequently, the languages accepted by
finite-memory automata exhibit properties
similar to those of regular languages and are,
thus, termed {\em quasi-regular} languages.

Automata over infinite alphabets have gained increasing importance
in computer science, as they provide formal models 
for systems that operate over unbounded data domains.
Infinite alphabets may be viewed as representing data,
such as
data values in XML documents~\cite{AlonMNSV03},
nonces in cryptographic protocols~\cite{KurtzKW07},
object identities~\cite{GrigoreDPT13},
or parameters of method calls~\cite{HowarJV19}.
The ability to reason about infinite alphabets while maintaining 
decidability of key properties makes these models particularly 
valuable for the formal verification of data-aware systems.

Finite-memory automata provide a structured way to reason
about such systems by focusing on patterns and repetitions
of data values rather than their specific identities.
This perspective is particularly useful in applications
where the exact values of data are less important 
than their relative behavior over time.
In particular,
quasi-regular languages form a subclass of
nominal languages
(also known as sets with atoms or Fraenkel-Mostowski sets)~\cite{Bojanczyk19},
and finite-memory automata themselves are expressively
equivalent to nominal automata
(also known as orbit-finite automata)~\cite{BojanczykKL14}.

A fundamental property of classical finite-state automata
is regularity, 
often characterized by the pumping lemma.
The pumping lemma is a key tool for proving
that certain languages are not regular by identifying 
structural constraints on words within the language. 
However, it is well known that the pumping lemma
does not extend to quasi-regular languages,
see~\cite[Example~3]{KaminskiF94},
where a deterministic automaton with two registers accepts
a language serving as a counterexample to the pumping lemma.
This leaves open, in the existing literature,
whether a similar counterexample exists already for one-register
automata.
One of the main results of this paper shows that it does not:
every one-register finite-memory automaton
satisfies the classical pumping lemma.

In general, the pumping lemma for quasi-regular languages
requires the pumped patterns to be generated by a finite-order
permutation.
In this variant, the pumping lemma depends on two constants:
the input length $p$ and the order of the permutation $d$.

For a regular language, the value of $p$
can always be chosen to be
the number of states of a finite automaton accepting the language,
regardless of whether it is deterministic or nondeterministic.
For finite-memory automata, this value
yields a finite-order permutation
of order at most
$(2r)!=O(e^{2r\log r})$~\cite[Theorem 7]{Danieli26},
where $r$ is the number of registers.

In this paper, we study the descriptional
complexity of pumping constants for several register-based
automata models.
Specifically, we derive upper bounds
on the least possible
pumping length $p$ and permutation-order $d$
in terms of the number of states and registers
of the automaton.
We consider one-register finite-memory automata,
multi-register finite-memory automata, and
hierarchical register automata.
Where possible, we also prove
lower bounds showing that these estimates
cannot be substantially improved.
The results are summarized in~\Cref{t: pl fma},
where $n$ is the number of states and $r$ is the number of registers.
Thus, our results exhibit a trade-off between the pumping length $p$ and
the permutation order $d$: short pumping length forces larger
permutation order, while larger pumping length allows the order to be
reduced, first to $2$ for general finite-memory automata and, in the
one-register and hierarchical cases, to $1$.

The table records only the positive upper bounds.
We also prove two complementary negative and algorithmic results.
First, for multi-register automata, it is undecidable
whether the recognized language satisfies the pumping lemma
for given constants $p$ and $d$;
it is also undecidable whether the classical pumping lemma
holds at all.
Thus, there is no general algorithm for computing optimal
pumping constants.
Second, for non-guessing one-register automata,
the situation improves: the minimal
classical pumping length is computable.

\begin{table}[h]
\centering
\renewcommand{\arraystretch}{1.4}
\begin{tabularx}{0.75\textwidth}{| c | Y | Y |}
\hline
\multicolumn{1}{|c|}{Model}
& Pumping length
& Permutation order \\
\hline
\multirow{2}{*}{$\FMA_1$}
& $n$
& $2$ \\
& $O(n^2)$
& $1$ \\
\hline
\multirow{2}{*}{$\FMA_r$}
& $n$
& $e^{(1+o(1))\sqrt{r\log r}}$ \\
& $O((r^2)!n)$
& $2$ \\
\hline
$\HRA_r$
& $O(n^{r+1})$
& $1$ \\
\hline
\end{tabularx}
\caption{Upper bounds for pumping constants proved in this paper.}
\label{t: pl fma}
\end{table}

\paragraph{Related work}

The constant in the pumping lemma for
regular languages
has been studied extensively
in terms of descriptional complexity (its relation
to the number of states/ transitions)~\cite{GruberHR26,GruberHR24},
computational complexity (the difficulty of computing
the minimal pumping constant for a given
regular language)~\cite{GruberHR26},
and operational complexity (how the minimal pumping constant changes
under Boolean and regular operations)~\cite{DassowI22,HolzerR23}.

The generalized pumping lemma
for quasi-regular languages
has appeared in several equivalent forms in the literature.
A formal presentation appears in~\cite[Theorem 3.1]{NakanishiTS23},
and the result was independently proved earlier by two
groups; see~\cite{Klin14} and~\cite[Lemma 46]{Veldhoven18}.
A refinement of this pumping lemma using finite-order
permutations was presented in~\cite[Theorem 7]{Danieli26}.

\paragraph{Organization.}
The results are presented in increasing order of technical difficulty.
After recalling finite-memory automata and the relevant pumping
definitions, we first treat general multi-register automata in
\Cref{s: fma r pl}.  There we prove two complementary positive
results: with the classical pumping length $p=n$, the required
permutation order can be bounded in terms of Landau's function; and, if the pumping length is increased by a constant depending
only on $r$, then permutations of order at most two always suffice.
We also prove the undecidability results there, showing that optimal
pumping constants cannot be computed in general.

We then turn to the more delicate one-register case in
\Cref{s: one register}.  This is the first point at which the
permutation can be eliminated completely: every one-register
finite-memory automaton satisfies the classical pumping lemma.  We also
show that the resulting quadratic pumping length is optimal up to
constant factors, and we prove that the minimal classical pumping
length is computable for non-guessing one-register automata.  Finally,
we extend the classical pumping result to hierarchical register
automata, whose proof combines the one-register argument with the
hierarchical structure of register updates.


\section{Finite-memory automata}

Throughout this paper, we employ the following conventions.
\begin{itemize}
\item
$\Sigma$ is a fixed infinite alphabet.
\item
Symbols in $\Sigma$ are denoted by $\sigma$,\,$\theta$, etc.,
sometimes indexed or primed.
\item
Bold lowercase Greek letters $\bsigma$, $\btheta$, etc.,
also sometimes indexed or primed,
denote words over $\Sigma$.
\item
Symbols that occur in a word denoted by a boldface letter
are always denoted by
the same \emph{non-boldface} letter with an appropriate subscript.
For example,
the letters that occur in $\bsigma^\prime$
are denoted by~$\sigma_i^\prime$.
\item
For a word
$\bsigma = \sigma_1 \sigma_2 \cdots \sigma_n \in \Sigma^\ast$,
we denote by $[\bsigma]$
the set of letters occurring in $\bsigma$:
\[
[\bsigma] = \{ \sigma_1 , \sigma_2 , \ldots , \sigma_n \}
\,
,
\]
and
call this set the {\em contents} of $\bsigma$.
\item
For a subset $\Omega$ of $\Sigma$ and a positive integer $r$,
we denote by $\Omega^{r_{\neq}}$ the subset of $\Omega^r$
whose elements do not contain repeated letters:
\[
\Omega^{r_{\neq}} =
\left\{
\sigma_1\ldots\sigma_{r} \in \Omega^{r}
:
\mbox{for all} \ i \ \mbox{and} \ j, \
\mbox{if} \ i \neq j \mbox{, then} \ \sigma_i \neq \sigma_j
\right\}
\,
.
\]
\item 
For a sequence of words
$\bsigma_1,\bsigma_2,\ldots,\bsigma_m$
we denote by $\prod_{i=1}^m\bsigma_i$,
the word which is the concatenation of
$\bsigma_i$'s in ascending order:
$\prod_{i=1}^m\bsigma_i = \bsigma_1\bsigma_2\cdots\bsigma_m.$
\end{itemize}

Next, we recall the definition of finite-memory automata.
Originally, finite-memory automata defined in~\cite{KaminskiF94}
were only allowed to store symbols that appeared so far;
later in~\cite{KaminskiZ08}, the model was extended with the ability
to non-deterministically reassign values to its registers.

\begin{definition}
\label{d: fma}
An \emph{$r$-register finite-memory automaton}
is a system
$\bA =
\langle S,s_0,F,D,\btheta_0,\mu \rangle$
whose components are as follows.
\begin{itemize}
\item $S,s_0\in S,$ and $ F\subseteq S$
are the finite set of \emph{states},
the \emph{initial state},
and the subset
of \emph{accepting states}, respectively.
\item $D \subset \Sigma$ is a finite set
of {\em distinguished} letters (constants).
\item $\btheta_0=
\tuple{\theta_{0,1},\theta_{0,2},\ldots,\theta_{0,r}}
\in (\Sigma\setminus D)^{r_{\neq}}$
is the initial register assignment,
where $\theta_{0,i}$ is the initial value of register $i$,
for $i=1,2,\ldots,r$.
\item $\mu$ is a finite set of transition
rules of the form $s\trans{\varphi}s^\prime$,
where $s,s^\prime\in S$
and the \emph{transition constraint}
$\varphi$
is a Boolean combination of equalities involving the variables
$x_1,x_2,\ldots,x_r,y,x_1^\prime,x_2^\prime,\ldots,x_r^\prime,$
and letters from $D$.
\end{itemize}
\end{definition}

Intuitively, $\varphi$ must
be satisfied by the current register contents
$x_1,x_2,\ldots,x_r$,
the input symbol $y$,
and the next register contents
$x_1^\prime,x_2^\prime,\ldots,x_r^\prime$.

The set $\Reg=(\Sigma\setminus D)^{r_{\neq}}$
contains all possible register assignments,
i.e., all $r$-tuples of pairwise distinct non-constants.
A \emph{configuration} $\conf s {\btheta} \in S\times \Reg$ of
$\bA$,
consists of a state $s\in S$ and register values
$\btheta = (\theta_1,\theta_2,\ldots, \theta_r) \in \Reg$,
where $\theta_i$ is the letter stored in the $i$th register.
If $s\in F$, it is called an \emph{accepting} configuration.
The configuration $\conf{s_0}{\btheta_0}$
is called the initial configuration.

A transition rule $s\trans{\varphi}s^\prime$ and
letters $\tau_1,\tau_2,\ldots,\tau_r$,
$\sigma$,
and $\tau_1^\prime,\tau_2^\prime,\ldots,\tau_r^\prime$
that satisfy $\varphi$
induce a transition
\begin{equation}
\label{eq: tran letter}
\conf{s}{\tuple{\tau_1,\tau_2,\ldots,\tau_r}}
\trans{\sigma}
\conf{s^\prime}{\tuple{\tau_1^\prime,\tau_2^\prime,\ldots,\tau_r^\prime}}
\, .
\end{equation}

A \emph{run} of $\bA$ on a word
$\bsigma=\sigma_1\sigma_2\cdots \sigma_m
\in \Sigma^\ast$ is a sequence of consecutive configurations
\begin{align} \label{eq:run}
{\conf {s_0} {\btheta_0}}
\trans{\sigma_1}
{\conf {s_1}{\btheta_1}}
\trans{\sigma_2}
\cdots 
\trans{\sigma_m} 
\conf{s_m}{\btheta_m}.
\end{align}
We shall also write
$\conf {s_0} {\btheta_0}
\transword{\bsigma} \conf{s_m}{\btheta_m}$.

A run is accepting
if it ends in an accepting configuration.
In such a case, we say that $\bsigma$ is accepted by $\bA$.
The set of all accepted words is denoted by $\langof{\bA}$
and is called the language of $\bA$.
Languages of finite-memory automata
are called
\emph{quasi-regular languages}.




\begin{example}[{\cite[Example 1]{KaminskiF94}}]
\label{e: two}
Consider a one-register finite-memory automaton $\bA$,
whose self-explanatory diagram is shown
in~\Cref{fig: aut l_eq},
where the symbol $\top$ denotes
the truth formula.
The language of $\bA$ consists
precisely of
all words over $\Sigma$
in which some letter appears more than once:
\[
L_{\text{eq}}:=\langof{\bA} = 
\{ \sigma_1 \sigma_2 \cdots \sigma_m \in\Sigma^\ast: \ \mbox{there exist} \
1 \leq i < j \leq m \ \mbox{such that} \ \sigma_i = \sigma_j \}\, .
\]

\begin{figure}[ht]
\begin{center}
\begin{tikzpicture}[scale=0.15]
\tikzstyle{every node}+=[inner sep=0pt]
\draw [black] (6.4,-28.21) -- (15.35,-28.29);
\fill [black] (15.35,-28.29) -- (14.55,-27.78) -- (14.55,-28.78);
\draw [black] (18.4,-28.2) circle (3);
\draw (18.4,-28.2) node {$s_0$};
\draw [black] (39.2,-28.3) circle (3);
\draw (39.2,-28.3) node {$s$};
\draw [black] (60.3,-28.2) circle (3);
\draw (60.3,-28.2) node {$f$};
\draw [black] (60.3,-28.2) circle (2.4);
\draw [black] (17.077,-25.52) arc (234:-54:2.25);
\draw (18.4,-20.95) node [above] {$\top$};
\fill [black] (19.72,-25.52) -- (20.6,-25.17) -- (19.79,-24.58);
\draw [black] (21.4,-28.21) -- (36.2,-28.29);
\fill [black] (36.2,-28.29) -- (35.4,-27.78) -- (35.4,-28.78);
\draw (28.8,-28.85) node [below] {$y=x_1^\prime$};
\draw [black] (42.2,-28.29) -- (57.3,-28.21);
\fill [black] (57.3,-28.21) -- (56.5,-27.72) -- (56.5,-28.72);
\draw (49.75,-30.05) node [below] {$x_1=y$};
\draw [black] (37.877,-25.62) arc (234:-54:2.25);
\draw (39.2,-21.05) node [above] {$x_1=x_1^\prime$};
\fill [black] (40.52,-25.62) -- (41.4,-25.27) -- (40.59,-24.68);
\draw [black] (58.977,-25.52) arc (234:-54:2.25);
\draw (60.3,-20.95) node [above] {$\top$};
\fill [black] (61.62,-25.52) -- (62.5,-25.17) -- (61.69,-24.58);
\end{tikzpicture}
\end{center}
\centering
\caption{The graph representation of $\bA$.}
\label{fig: aut l_eq}
\end{figure}
\end{example}

\begin{example}[{\cite[Example 4]{KaminskiZ08}}]
\label{e: two and hald}
The language
\[
L_{\mathrm{last}}=
\{ \sigma_1 \sigma_2 \cdots \sigma_m \in \Sigma^\ast :
m\geq 2,\ \mbox{and for all} \
1 \leq i  < m, \ \sigma_i \neq \sigma_m \}\, .
\]
consists of all words in which the last letter does not appear
earlier and it is accepted by the automaton $\bB$
shown in Figure~\ref{fig: auto l last}.

\begin{figure}[H]
\centering
\begin{tikzpicture}[scale=0.15]
\tikzstyle{every node}+=[inner sep=0pt]
\draw [black] (6.4,-28.21) -- (15.35,-28.29);
\fill [black] (15.35,-28.29) -- (14.55,-27.78) -- (14.55,-28.78);
\draw [black] (18.4,-28.2) circle (3);
\draw (18.4,-28.2) node {$s_0$};
\draw [black] (39.2,-28.3) circle (3);
\draw (39.2,-28.3) node {$s$};
\draw [black] (60.3,-28.2) circle (3);
\draw (60.3,-28.2) node {$f$};
\draw [black] (60.3,-28.2) circle (2.4);
\draw [black] (21.4,-28.21) -- (36.2,-28.29);
\fill [black] (36.2,-28.29) -- (35.4,-27.78) -- (35.4,-28.78);
\draw (28.8,-30.75) node [below] {$y\neq x_1^\prime$};
\draw [black] (42.2,-28.29) -- (57.3,-28.21);
\fill [black] (57.3,-28.21) -- (56.5,-27.72) -- (56.5,-28.72);
\draw (49.75,-31.75) node [below] {$x_1=y$};
\draw [black] (37.877,-25.62) arc (234:-54:2.25);
\draw (39.2,-21.05) node [above] {$x_1=x_1^\prime\neq y$};
\fill [black] (40.52,-25.62) -- (41.4,-25.27) -- (40.59,-24.68);
\end{tikzpicture}
\caption{The graph representation of $\bB$.}
\label{fig: auto l last}
\end{figure}
\end{example}

\begin{remark}
\label{r: guessing}
By Definition~\ref{d: fma},
an automaton may reassign register values 
non-deterministically.
Such transitions, often called \emph{guesses},
strictly increase the expressive
power of finite-memory automata~\cite{ChengK98,KaminskiZ08},
e.g., $L_{\text{last}}$ from~\Cref{e: two and hald}
cannot be recognized by finite-memory automata without guesses.
\end{remark}

\begin{definition}
A finite-memory automaton is called \emph{non-guessing}
if every constraint $\varphi$ appearing in a transition rule
forces every ``post-register'' value to be
equal to some ``pre-register'' value or the current symbol.
\end{definition}

Recall the pumping lemma for regular languages.

\begin{lemma}[{\cite[Lemma 8]{RabinS59}}]
\label{l: pl}
Let $L$ be a regular language over a finite alphabet.
Then there is a constant $p$ for which every word $\bsigma\in L$,
with $|\bsigma|\geq p$, can be decomposed as 
$\bsigma=\btau\bupsilon\bvarphi$ such that
\begin{enumerate}
\item $|\bupsilon|>0$, $|\btau\bupsilon|\leq p$, and
\item for all $k=0,1,2,\ldots,$
\begin{equation}
\btau\bupsilon^k\bvarphi\in L\, .
\end{equation}
\end{enumerate}
\end{lemma}

\begin{example}
The language $L_{\text{eq}}$ from~\Cref{e: two}
satisfies the pumping lemma for $p=3$.
Indeed, for $\bsigma\in L_{\text{eq}}$,
$\bupsilon$ can be the first symbol not belonging to
a pair of matching symbols.
Moreover, $L_{\text{eq}}$ does not satisfy the pumping lemma for $p=2$,
because the word $\sigma\sigma$ cannot be pumped for $k=0$.

The language $L_{\text{last}}$ from~\Cref{e: two and hald}
also satisfies the pumping lemma for $p=3$, but not $p=2$.
Indeed, every decomposition in which $\bupsilon$ is
not the last symbol works
for words of length at least three,
however, the word $\sigma\tau$ for $\sigma\neq\tau$
cannot be pumped for $k=0$.
\end{example}

\begin{example}
\label{e: three}
Let $\bC$ be a two-register finite-memory automaton
depicted in~\Cref{fig: aut l_no-pl}.
For readability, whenever the post-transition 
register variables, $x_1^\prime,x_2^\prime$,
are absent from a constraint, they are assumed to be preserved,
i.e., $x_1^\prime=x_1$ and $x_2^\prime=x_2$.

\begin{figure}[H]
\begin{center}
\begin{tikzpicture}[scale=0.15]
\tikzstyle{every node}+=[inner sep=0pt]
\draw [black] (24.2,-19) circle (3);
\draw (24.2,-19) node {$s_0$};
\draw [black] (39.9,-19) circle (3);
\draw (39.9,-19) node {$s_1$};
\draw [black] (55.8,-19) circle (3);
\draw (55.8,-19) node {$s_2$};
\draw [black] (55.8,-37.4) circle (3);
\draw (55.8,-37.4) node {$s_3$};
\draw [black] (39.9,-37.4) circle (3);
\draw (39.9,-37.4) node {$s_4$};
\draw [black] (24.2,-37.4) circle (3);
\draw (24.2,-37.4) node {$s_5$};
\draw [black] (24.2,-37.4) circle (2.4);
\draw [black] (12,-19) -- (21.2,-19);
\fill [black] (21.2,-19) -- (20.4,-18.5) -- (20.4,-19.5);
\draw [black] (42.9,-19) -- (52.8,-19);
\fill [black] (52.8,-19) -- (52,-18.5) -- (52,-19.5);
\draw (47.85,-19.5) node [below] {$y=x_2^\prime$};
\draw [black] (27.2,-19) -- (36.9,-19);
\fill [black] (36.9,-19) -- (36.1,-18.5) -- (36.1,-19.5);
\draw (32.05,-19.5) node [below] {$y=x_1^\prime$};
\draw [black] (55.8,-22) -- (55.8,-34.4);
\fill [black] (55.8,-34.4) -- (56.3,-33.6) -- (55.3,-33.6);
\draw (57.3,-28.2) node [right] {$x_1=y$};
\draw [black] (52.8,-37.4) -- (42.9,-37.4);
\fill [black] (42.9,-37.4) -- (43.7,-37.9) -- (43.7,-36.9);
\draw (47.85,-36.9) node [above] {$y=x_1^\prime$};
\draw [black] (39.9,-34.4) -- (39.9,-22);
\fill [black] (39.9,-22) -- (39.4,-22.8) -- (40.4,-22.8);
\draw (40.4,-28.2) node [right] {$x_2=y$};
\draw [black] (37.95,-21.28) -- (26.15,-35.12);
\fill [black] (26.15,-35.12) -- (27.05,-34.83) -- (26.29,-34.18);
\draw (31.5,-26.76) node [left] {$x_1=y$};
\end{tikzpicture}
\end{center}
\caption{The graph representation of $\bC$.}
\label{fig: aut l_no-pl}
\end{figure}

The automaton $\bC$ impose overlapping
equality constraints, see~\Cref{fig: forms}.

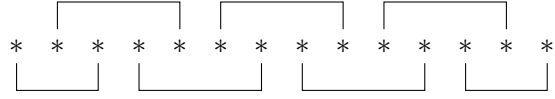
\begin{figure}[htbp]
\centering
\begin{tikzpicture}[scale=0.9][
thick,
star node/.style={font=\LARGE\bfseries, inner sep=2pt}
]

\foreach \i in {1,...,14} {
\node (n\i) at (\i*0.6, 0) {$*$};
}

\draw (n2.north)  -- ++(0, 0.4) -| (n5.north);
\draw (n6.north)  -- ++(0, 0.4) -| (n9.north);
\draw (n10.north) -- ++(0, 0.4) -| (n13.north);

\draw (n1.south)  -- ++(0,-0.4) -| (n3.south);
\draw (n4.south)  -- ++(0,-0.4) -| (n7.south);
\draw (n8.south)  -- ++(0,-0.4) -| (n11.south);
\draw (n12.south) -- ++(0,-0.4) -| (n14.south);

\end{tikzpicture}
\caption{Equality constraints for words in $L(\bC)$.}
\label{fig: forms}
\end{figure}

As a result, the language accepted
by this automaton
$L_{\text{no-pl}}:=\langof{\bC}$ does not 
satisfy~\Cref{l: pl} for any $p$~\cite[Example 3]{KaminskiF94}.

\end{example}

The pumping lemmas in this paper involve the following definitions.

\begin{definition}
\label{d: permutation}
The order of a permutation $\alpha : \Sigma \rightarrow \Sigma$,
$\operatorname{ord}(\alpha)$,
is the smallest positive integer $k$ such that 
$\alpha^k$ is the identity permutation:
$\alpha^{k} = \id$,
\footnote{
As usual, the product of two permutations is their composition
and powers are defined recursively by
$\alpha^k = \alpha \circ \alpha^{k-1}$.}
if no such $k$ exists, the order of $\alpha$
is infinite.
Moreover, we say that $\alpha$ is a $D$-{\em permutation}
if it is invariant on $D$
(i.e., for all $d\in D,\alpha(d)=d$).
\end{definition}

\begin{definition}
\label{d: extension}
Let $\alpha$ be a permutation of $\Sigma$.
We extend $\alpha $ to configurations by
$\alpha(s,\btheta) = (s,\alpha(\btheta))$.
\end{definition}

\begin{proposition}[Equivariance of FMA~{\cite[Lemma 1]{KaminskiF94}}]
\label{p: permutation word fma}
Let $\alpha$ be a $D$-permutation of $\Sigma$,
$\bsigma \in \Sigma^\ast$, and
let $c$ and $c^\prime$ be configurations of $\bA$ such that
$c \overset{\bsigma}{\longrightarrow} c^\prime$.
Then
$\alpha(c)
\overset{\alpha(\bsigma)}{\longrightarrow}
\alpha(c^\prime)$.
\end{proposition}


\section{Pumping lemmas for finite-memory automata}
\label{s: fma r pl}
Due to the close connection with orbit-finite theory
and symmetries, most generalizations from finite alphabets
to infinite alphabets use permutations
on the alphabet~\cite{Bojanczyk19,BojanczykKL14}.
The following definition is the generalized variant
of the pumping lemma for infinite alphabets.

\begin{definition}
\label{d: pumpable word}
Let $L\subseteq\Sigma^\ast$ and $\bsigma\in L$,
we say that $\bsigma$ is \emph{pumpable} for input length $p$
and permutation order $d$ with respect to $L$,
in short $(p,d)$\emph{-pumpable} when $L$ is clear from context,
if it satisfies
the following condition:

There is a decomposition $\bsigma=\btau\bupsilon\bvarphi$
and a $D$-permutation $\alpha$ of order at most $d$
such that:
\begin{enumerate}
\item $\size{\bupsilon}>0$ and $\size{\btau\bupsilon}\leq p$;
\item for all $k=1,2,\ldots,$
\begin{equation}
\label{eq: k fma}
\btau\bupsilon\alpha(\bupsilon)\alpha^2(\bupsilon)\cdots
\alpha^k(\bupsilon)\alpha^k(\bvarphi)\in L;
\end{equation}
and
\item
\begin{equation}
\label{eq: zero}
\btau\alpha^{-1}(\bvarphi)\in L.
\end{equation}
\end{enumerate}
\end{definition}

\begin{definition}
\label{d: pumpable language}
A language $L\subseteq\Sigma^\ast$
is called \emph{pumpable} for input length $p$
and permutation order $d$,
or $(p,d)$\emph{-pumpable},
if all $\bsigma\in L$ with $\size{\bsigma}\geq p$
are $(p,d)$-pumpable with respect to $L$.

We say that a language satisfies the \emph{classical} pumping lemma
if it is $(p,1)$-pumpable for some $p$.
\end{definition}

\subsection{Number of states as input length}

In~\cite[Theorem 7]{Danieli26}, it was shown
that the language of a
finite-memory automaton with $n$ states
and $r$ registers is $(n,(2r)!)$-pumpable.
In the proof it is noted
that the order of the permutation is actually
at most $g(2r)$,
where $g$ is Landau's function~--
the largest order of an element of
the symmetric group~\cite{Landau03}.
Equivalently, it is the maximal least common multiple
of positive integer partitions,
\[
g(r)=
\max_{\lambda_1+\cdots+\lambda_k=r}
\operatorname{lcm}(\lambda_1,\ldots,\lambda_k).
\]
We improve this bound.

\begin{theorem}
\label{t: pl one}
Let \(\bA\) be a finite-memory automaton with \(n\) states and \(r\)
registers.
\begin{enumerate}
\item If \(1\leq r\leq 6\), then \(L(\bA)\) is
\((n,r+1)\)-pumpable.
\item If \(r\geq 7\), then \(L(\bA)\) is
\((n,g(r))\)-pumpable.
\end{enumerate}
\end{theorem}

We use the following standard pumping observation.

\begin{lemma}
\label{l: pump}
Suppose that an accepting run on
\(\bsigma=\btau\bupsilon\bvarphi\) has the form
\[
c_0\transword{\btau}c_1
\transword{\bupsilon}c_2
\transword{\bvarphi}c_f,
\]
and that a finite-order \(D\)-permutation \(\alpha\) satisfies
\(\alpha(c_1)=c_2\).  Then the loop labelled by \(\bupsilon\) can be
pumped as in~\eqref{eq: k fma} and removed as in~\eqref{eq: zero}.
\end{lemma}

\begin{proof}
By equivariance~(\Cref{p: permutation word fma}),
\[
\alpha^j(c_1)
\transword{\alpha^j(\bupsilon)}
\alpha^{j+1}(c_1)
\qquad (j\geq 0).
\]
These runs may be concatenated and followed by the image under
\(\alpha^k\) of the suffix run.  Since permutations preserve states,
the resulting final configuration is accepting.  Applying
\(\alpha^{-1}\) to the suffix run gives the case in which the loop is
removed.
\end{proof}

The main point is therefore to bound the order of a permutation
mapping one register assignment to another.
For that, we shall need the following
result on special integer partitions.

\begin{lemma}
\label{l: closure}
Let $t,p\geq 0$ and $c_1,c_2,\ldots,c_p,a_1,a_2,\ldots,a_t$
be positive integers such that
\[
c_1+\cdots+c_p+a_1+\cdots+a_t=r\, .
\]
Then there are positive integers $b_1,b_2,\ldots,b_t$ with \(b_j>a_j\) for $j=1,2,\ldots,t$ with
\[
\operatorname{lcm}
(c_1,\ldots,c_p,b_1,\ldots,b_t)
\leq \max\{g(r),r+1\}.
\]
\end{lemma}

\begin{proof}
If \(t=0\), the displayed least common multiple is at most \(g(r)\).
If \(p=0\), let \(B=\max_j a_j\) and choose every \(b_j=B+1\).
The resulting least common multiple is \(B+1\leq r+1\).

Assume now that \(p,t>0\), and again let \(B=\max_j a_j\).  Let \(b\)
be the smallest multiple of \(c_1\) larger than \(B\).  Then
\[
B<b\leq B+c_1.
\]
Choose every \(b_j=b\).  Since \(c_1\mid b\),
\[
\operatorname{lcm}(c_1,\ldots,c_p,b)
=
\operatorname{lcm}(c_2,\ldots,c_p,b).
\]
Moreover,
\[
b+c_2+\cdots+c_p
\leq B+c_1+\cdots+c_p
\leq r.
\]
Thus \(b,c_2,\ldots,c_p\), together with sufficiently many \(1\)'s,
form a partition of \(r\).  Hence their least common multiple is at
most \(g(r)\).
\end{proof}

\begin{lemma}
\label{l: perm r}
For any
\[
\btheta,\btheta'\in(\Sigma\setminus D)^{r_{\neq}},
\]
there is a \(D\)-permutation \(\alpha\) such that
\[
\alpha(\btheta)=\btheta'
\qquad\text{and}\qquad
\operatorname{ord}(\alpha)
\leq \max\{g(r),r+1\}.
\]
\end{lemma}

\begin{proof}
Form the directed graph with vertices $[\btheta]\cup [\btheta^\prime]$ and edges
\[
\theta_i\longrightarrow\theta_i',
\qquad i=1,\ldots,r.
\]
Every vertex has indegree and outdegree at most one, so each component
is either a directed cycle or a directed path.
Let \(c_1,\ldots,c_p\) be the cycle lengths and
\(a_1,\ldots,a_t\) the path lengths.  Since the graph has \(r\) edges,
\[
c_1+\cdots+c_p+a_1+\cdots+a_t=r.
\]
A path of length \(a_j\) can be completed, using fresh letters, to a
cycle of any length \(b_j>a_j\).  The resulting cycles define a
\(D\)-permutation mapping every \(\theta_i\) to \(\theta_i'\), and its
order is
\[
\operatorname{lcm}(c_1,\ldots,c_p,b_1,\ldots,b_t).
\]
The result follows from \Cref{l: closure}.
\end{proof}

Finally, we compare \(g(r)\) with \(r+1\).

\begin{lemma}
\label{l: gr r2}
For every \(r\geq 7\),
\[
g(r)\geq r+1.
\]
\end{lemma}

\begin{proof}
If \(r=2m+1\), then \(m\geq 3\), and the partition
\[
r=m+(m+1)
\]
gives
\[
g(r)\geq \operatorname{lcm}(m,m+1)=m(m+1)\geq 2m+2=r+1.
\]

If \(r=8\), the partition \(8=3+5\) gives \(g(8)\geq 15>9\).

Finally, let \(r=2m\geq 10\).  Since
\[
2m=(m-1)+(m+1)
\]
and
\[
\gcd(m-1,m+1)\leq 2,
\]
we have
\[
g(2m)
\geq \operatorname{lcm}(m-1,m+1)
\geq \frac{m^2-1}{2}
\geq 2m+1.
\]
\end{proof}

\begin{proof}[Proof of \Cref{t: pl one}]
In every sufficiently long accepting run, two of the first \(n+1\)
configurations have the same state.  Let their register assignments be
\(\btheta\) and \(\btheta'\).  By \Cref{l: perm r}, there is a
\(D\)-permutation \(\alpha\) mapping the first configuration to the
second and satisfying
\[
\operatorname{ord}(\alpha)
\leq \max\{g(r),r+1\}.
\]
The pumping property follows from \Cref{l: pump}.  Thus \(L(\bA)\) is
\[
\bigl(n,\max\{g(r),r+1\}\bigr)\text{-pumpable}.
\]

For \(1\leq r\leq 6\),
\[
g(1),\ldots,g(6)=1,2,3,4,6,6,
\]
so \(\max\{g(r),r+1\}=r+1\).  For \(r\geq 7\),
\Cref{l: gr r2} gives
\(\max\{g(r),r+1\}=g(r)\).
\end{proof}

\Cref{t: pl one}
provides upper bounds for the order of permutations
in terms of descriptional complexity of the number
of states and registers.
As a matter of fact, these bounds are optimal in the worst-case.

\begin{example}
\label{e: lpi}
Let $\pi:\{1,2,\ldots,r\}\to\{1,2,\ldots,r\}$
be a permutation of order $g(r)$
and let ${\tau_1,\tau_2,\ldots,\tau_r}$
be distinct symbols from ${\Sigma\setminus\{\#\}}$.
Consider the language
\[
L_\pi=\{\#\bsigma_0\#\bsigma_1\#\cdots\#\bsigma_k:
k\geq 0
\mbox{ and }
\bsigma_i=\tau_{\pi^i(1)}\tau_{\pi^i(2)}\cdots\tau_{\pi^i(r)}
\mbox{ for }i=0,1,2,\ldots,k\}\, .
\]
It is accepted by a finite-memory automaton
with $n=r+1$ states and $r$ registers.
We contend that for every $d<g(r)$,
$L_\pi$ is not $(n,d)$-pumpable.
Assume, to the contrary, that it is and
consider the word $\#\tau_1\tau_2\cdots\tau_r\#$
of length $r+2>n$.
In any pumping by means of a finite-order permutation $\alpha$,
the subword being pumped must be
$\#\tau_1\tau_2\cdots\tau_r$ or $\tau_1\tau_2\cdots\tau_r\#$.
Therefore, $\alpha$ must agree with $\pi$
on $\tau_1,\tau_2,\ldots,\tau_r$,
i.e., $\alpha(\tau_i)=\tau_{\pi(i)}$.
By the definition of $\pi$,
the order of $\alpha$ is at least $g(r)$.
\end{example}

\subsection{Input lengths with small permutation orders}

By allowing a larger pumping length,
we can obtain permutations of smaller order.
In fact, we show that for sufficiently long words,
we can always find a permutation of order at most $2$.

\begin{theorem}
\label{t: pl two}
For every $r\geq 1$,
every language accepted by a finite-memory automaton
with $n$ states and $r$ registers is $(3(r^2!)n,2)$-pumpable.
\end{theorem}

The proof of~\Cref{t: pl two}
uses the fact that an accepting run over a word
of length at least $3(r^2!)n$ contains
more than $3(r^2!)$ configurations with the same state.
The following lemma shows that among any such sequence
of register assignments, two can be mapped to one another.

\begin{lemma}
\label{l: ramsey}
Let $\btheta_1,\ldots,\btheta_N$
be a sequence of elements of
$(\Sigma\setminus D)^{r_{\neq}}$.
If $N>3 \left( r^2!\right) $,
there are non-negative integers $1\leq i<j\leq N$
and a $D$-permutation $\alpha$
of order at most $2$ such that
$\alpha(\btheta_i)=\btheta_j$.
\end{lemma}

\begin{proof}
First, we show
a sufficient condition for a permutation of order $2$
to map one register assignment to another.
Let $\btheta,\btheta^\prime\in (\Sigma\setminus D)^{r_{\neq}}$.
If $\theta_\ell\notin [\btheta^\prime]\setminus \{\theta_\ell^\prime\}$
for all $\ell=1,2,\ldots,r$,
then the permutation $\alpha$,
defined by $\alpha(\theta_\ell)=\theta^\prime_\ell$ and
$\alpha(\theta_\ell^\prime)=\theta_\ell$ for $\ell=1,2,\ldots,r$,
is a $D$-permutation
of order at most two.

We contend that there are $1\leq i<j\leq N$,
such that $\btheta_i$ and $\btheta_j$
satisfy this condition.
Assume to the contrary,
that for every $1\leq i<j\leq N$,
$\btheta_i,\btheta_j$ do not satisfy the mentioned condition.
Let $\ell_{i,j}\neq t_{i,j}$ be such that
$\theta_{i,\ell_{i,j}}=\theta_{j,t_{i,j}}$.
Let $X=\{(\ell,t):\ell,t\in\{1,2,\ldots,r\},\ell\neq t\}$.
Note that $|X|=r(r-1)< r^2$.

Consider the complete graph of $N$ vertices $\{ v_1,v_2,\ldots,v_N \}$,
with coloring of edge $\{ v_i,v_j \}$
with the color $(\ell_{i,j},t_{i,j})\in X$.
By the \emph{multicolor Ramsey theorem}~\cite[Theorem 6]{Greenwood55},
for $N> 3 (|X|!)$, there is a monochromatic triangle in the graph.
Let $v_i$, $v_j$, and $v_k$, $i<j<k$, be the vertices
of that triangle.
Let $(\ell,t)\in X$ be the color of the edges $(v_i,v_j),(v_i,v_k)$,
and $(v_j,v_k)$.
By the definition of $\ell$ and $t$,
$\theta_{i,\ell}=\theta_{j,t}$,
$\theta_{i,\ell}=\theta_{k,t}$,
and
$\theta_{j,\ell}=\theta_{k,t}$.
In particular,
$\theta_{j,\ell}=\theta_{j,t}$,
but $\ell\neq t$.
This contradiction proves our contention.
\end{proof}

\subsection{Minimal pumping constants}

Unfortunately, trying to optimize and perhaps
compute the minimal constants for a given
finite-memory automaton is impossible
due to the following theorem.

\begin{theorem}
\label{t: ud prob}
The following decision problems are undecidable for finite-memory
automata.
\begin{itemize}
\item For given $p,d$ and $\textsc{FMA}$ $\bA$, is $L(\bA)$ $(p,d)$-pumpable?
\item For a given $p$ and $\textsc{FMA}$ $\bA$, is
$L(\bA)$ $(p,d)$-pumpable for some $d$?
\item For $\textsc{FMA}$ $\bA$, does the classical pumping lemma hold for $L(\bA)$?
\end{itemize}
\end{theorem}

\begin{proof}
We reduce the universality problem, which is known
to be undecidable for finite-memory 
automata~\cite[Theorem 5.1]{NevenSV04},
to our problem.
Let $\bB$ be a finite-memory automaton
for which we want to decide universality.
Fix $p$, let $\# \notin \Sigma$ and
consider the following language,
\begin{equation}
\label{eq: reduction}
L_{p,B}:=\Sigma^{p} \cdot \# \cdot \Sigma^\ast \ \cup \ \Sigma^\ast \cdot \#\cdot L(\bB)
\ \cup \ \#^\ast\cdot\Sigma^\ast \, .
\end{equation}
$L_{p,B}$ is a quasi-regular language, because quasi-regular languages
are closed under union and concatenation.\footnote{
This language is similar to the one in~\cite{GruberHR24}.}
Let $\bA$ be a finite-memory automaton
that accepts $L_{p,B}$.
Note that $\Sigma^{p}$ is not pumpable for input length $p$
and any permutation order $d$.

We distinguish between the cases of $L(\bB)=\Sigma^\ast$ and $L(\bB) \neq \Sigma^\ast$.
\begin{itemize}
\item If $L(\bB)=\Sigma^\ast$, then
$L(\bA)=\Sigma^\ast \cdot \#\cdot
\Sigma^\ast \ \cup \ \#^\ast\cdot\Sigma^\ast $,
that is $(1,1)$-pumpable and, therefore, is also $(p,d)$-pumpable.
\item If $L(\bB)\neq\Sigma^\ast$,
let $\bomega\notin L(\bB)$ and $\bsigma\in \Sigma^{p}$.
Since $\bsigma\#\bomega$ belongs to the first component of~\eqref{eq: reduction},
it is in $L(\bA)$.
We contend that $\bsigma\#\bomega$
is not $(p,d)$-pumpable with respect to $L_{p,B}$.
Assume, to the contrary, that there
is a decomposition $\btau\bupsilon\bvarphi$
of $\bsigma$ and a finite-order permutation $\alpha$
of order $d_0\leq d$
that satisfies the pumping conditions.
However, for $k =d_0$ in~\eqref{eq: k fma}
\[
\btau\bupsilon\alpha(\bupsilon)\cdots \alpha^{k-1}(\bupsilon)
\bupsilon \bvarphi \# \bomega
\]
is not in $L_{p,B}=L(\bA)$.
\end{itemize}

Therefore, $L(\bB)=\Sigma^\ast$ if and only if $L(\bA)$
is $(p,d)$-pumpable, implying that
the first problem is undecidable.

For the second problem, notice that for the same reduction
$L(\bB)=\Sigma^\ast$ if and only if $L(\bA)$
is $(p,d)$-pumpable for some $d$.

For the third problem, 
replace $\Sigma^{p}$ with the language $L_{\text{no-pl}}$
from~\Cref{e: three}, that is not pumpable without permutations.
Thus, $L(\bB)=\Sigma^\ast$ if and only if $L(\bA)$
is $(p,1)$-pumpable for some $p$.
\end{proof}

\section{Pumping lemma for one-register automata}
\label{s: one register}
In this section,
we show that, for $r = 1$,
there is a permutation of order one.
In other words, the classical pumping lemma 
holds for
one-register quasi-regular languages.

\begin{theorem}
\label{t: pl one fma}
Let $\bA$ be a one-register finite-memory automaton,
then $L(\bA)$ is $(11n^2+n,1)$-pumpable,
where $n$ is the number of states in $\bA$.
\end{theorem}

For the proof of~\Cref{t: pl one fma},
we find sufficient conditions
for a word to be pumped without a permutation.
Note that if an accepting run repeats configurations,
this is immediate.

\begin{lemma}
\label{l: same cfg}
Let $\bc=c_0,c_1,\ldots,c_m$
be an accepting run
of $\bA$ on $\bsigma=\sigma_1\cdots\sigma_m$.
If there are integers $0\leq i<j\leq m$
with $c_i=c_j$,
then $\bsigma$ can be pumped for the decomposition
$\btau=\sigma_1\cdots\sigma_{i}\,$,
$\bupsilon=\sigma_{i+1}\cdots\sigma_j\,$,
and $\bvarphi=\sigma_{j+1}\cdots\sigma_m\,$.
\end{lemma}

However,
over an infinite alphabet,
there is no guarantee that there is an accepting run with
repeating configurations.
Nevertheless, for long computations,
either the register value does not change
(thus, due to the finite number of states,
there is a repetition of configurations),
or there are multiple transitions in which
the register changes its content.

\begin{definition}
\label{d: register}
A transition $\conf{s}{\sigma_1}\trans{\sigma_2}\conf{s^\prime}{\sigma_3}$ 
is called \emph{register-preserving} if $\sigma_1=\sigma_3$
and \emph{register-updating} otherwise.
\end{definition}

For register-updating transitions
there are three cases for equalities between
the symbols $\sigma_1,\sigma_2,\sigma_3$:
\begin{itemize}
\item[(a)] $\sigma_1\neq \sigma_2=\sigma_3 (\neq \sigma_1)$.
\item[(b)] $\sigma_1= \sigma_2\neq\sigma_3 (\neq\sigma_1)$.
\item[(c)] $\sigma_1\neq \sigma_2\neq\sigma_3 (\neq\sigma_1)$.
\end{itemize}
Similarly, a triple $(\sigma_1,\sigma_2,\sigma_3)\in\Sigma^3$
with $\sigma_1\neq\sigma_3$,
is classified into
its corresponding type, either (a), (b), or (c).

\begin{remark}
\label{r: tran guess}
Transitions of type (b) and type (c) are not possible
in non-guessing finite-memory automata.
\end{remark}

\begin{definition}
\label{d: comp}
Two register-updating transitions:
$\conf{s}{\sigma_1}\trans{\sigma_2}\conf{s^\prime}{\sigma_3}$
and
$\conf{t}{\tau_1}\trans{\tau_2}\conf{t^\prime}{\tau_3}$
are called \emph{compatible} if one of the following holds:
\begin{itemize}
\item[(a)] Both are of type (a), $s=t$,
$\sigma_1\neq \tau_2=\tau_3 (\neq\sigma_1)$,
and $\tau_1\neq\sigma_2=\sigma_3 (\neq \tau_1)$.
\item[(b)] Both are of type (b), $s^\prime=t^\prime$,
$\sigma_1=\sigma_2\neq \tau_3 (\neq \sigma_1)$, and
$\tau_1=\tau_2\neq \sigma_3 (\neq \tau_1)$.
\item[(c)] Both are of type (c), $s^\prime=t^\prime$,
$\sigma_1\neq\sigma_2\neq \tau_3 (\neq \sigma_1)$, and
$\tau_1\neq\tau_2\neq \sigma_3 (\neq \tau_1)$.
\end{itemize}

Similarly, two triples $(\sigma_1,\sigma_2,\sigma_3),(\tau_1,\tau_2,\tau_3)$
with $\sigma_1\neq\sigma_3$ and $\tau_1\neq \tau_3$
are called \emph{compatible}
if they have the same type and satisfy the same constraints as above (all except the state constraint).
\end{definition}

If an accepting run contains two compatible transitions,
then we have straightforward pumping.

\begin{lemma}
\label{l: diff registers}
Let $\bc=c_0,c_1,\ldots,c_m$
be an accepting run of $\bA$ on $\bsigma=\sigma_1\cdots\sigma_m$,
$c_i =\conf{s_i}{\theta_i}$, $i = 0,1,\ldots,m$.
If there are integers $0\leq i<j< m$
such that $c_i\trans{\sigma_{i+1}}c_{i+1}$ and $c_j\trans{\sigma_{j+1}}c_{j+1}$
are compatible,
then, $\bsigma$ can be pumped for the decomposition
$\btau=\sigma_1\cdots\sigma_{i}\,$,
$\bupsilon=\sigma_{i+1}\cdots\sigma_j\,$,
and $\bvarphi=\sigma_{j+1}\cdots\sigma_m\,$.
\end{lemma}

\begin{proof}
For the case of compatibility of type (a), we have:
\[
\conf{s_o}{\theta_0}\transword{\sigma_1\sigma_2\cdots \sigma_i}
\conf{s_i}{\theta_i}\trans{\sigma_{i+1}}\conf{s_{i+1}}{\theta_{i+1}}
\transword{\sigma_{i+2}\sigma_{i+3}\cdots \sigma_j}
\conf{s_j}{\theta_j}\trans{\sigma_{j+1}}\conf{s_{j+1}}{\theta_{j+1}}
\transword{\sigma_{j+2}\sigma_{j+3}\cdots \sigma_m}
\conf{s_m}{\theta_m}\, .
\]
That is, $s_i = s_j$, $\theta_i \neq \sigma_{j+1}=\theta_{j+1}(\neq \theta_i)$,
and
$\theta_j\neq \sigma_{i+1}=\theta_{i+1}(\neq \theta_j)$.
We obtain the following shortened run:
\[
\conf{s_o}{\theta_0}\transword{\sigma_1\sigma_2\cdots \sigma_i}
\conf{s_i}{\theta_i}=
\conf{s_j}{\theta_i}\trans{\sigma_{j+1}}\conf{s_{j+1}}{\theta_{j+1}}
\transword{\sigma_{j+2}\sigma_{j+3}\cdots \sigma_m}
\conf{s_m}{\theta_m}\, ,
\]
and the extended runs:
\[
\conf{s_o}{\theta_0}\transword{\sigma_1\sigma_2\cdots \sigma_j}
\conf{s_j}{\theta_j}=\conf{s_i}{\theta_j}
\trans{\sigma_{i+1}}\conf{s_{i+1}}{\theta_{i+1}}
\transword{\sigma_{i+2}\sigma_{i+3}\cdots \sigma_m}
\conf{s_m}{\theta_m}\, .
\]
The other types are similar.
\end{proof}

Finally, we show that in a constant
number of triples of symbols,
there is always a compatible pair.

\begin{proposition}
\label{p: comp tri}
Let $(\sigma_{1,i},\sigma_{2,i},\sigma_{3,i})_{i=1}^N$
be a sequence of triples of symbols all of the same type.
\begin{itemize}
\item If the type is (a) or (b) and $N\geq 4$, there is a compatible pair.
\item If the type is (c) and $N\geq 6$, there is a compatible pair.
\end{itemize}
\end{proposition}

\begin{proof}
For type (a),
$\sigma_{1,i}\neq \sigma_{2,i}=\sigma_{3,i} (\neq \sigma_{1,i}), $ for all $i$.
We look for a pair $i<j$ such that $\sigma_{1,i}\neq\sigma_{2,j}$
and $\sigma_{1,j}\neq \sigma_{2,i}$.
Assume to the contrary that no such pair exists.
That is, for every pair of different $i$ and $j$,
either $\sigma_{1,i}=\sigma_{2,j}$ or $\sigma_{1,j}= \sigma_{2,i}$.
With each triple $(\sigma_{1,i},\sigma_{2,i},\sigma_{3,i})$, $i=1,2,3,4$,
we associate a counter $c_i$.
We initialize the counters $c_1=c_2=c_3=c_4=0$ and
for each pair $\{i,j\}$ increase $c_i$ by one if $\sigma_{1,i}=\sigma_{2,j}$
and increase $c_j$ by one if $\sigma_{1,j}= \sigma_{2,i}$.
Since there are $\binom{4}{2}=6$ such pairs, at the end $\sum_{i=1}^4c_i\geq 6$.
By the pigeonhole principle, there is an $i$ with $c_i\geq 2$.
Therefore, there are $j$ and $j^\prime$ such that $j\neq j^\prime$
and $\sigma_{1,i}=\sigma_{2,j}=\sigma_{2,j^\prime}$.
It follows that those $j$ and $j^\prime$ satisfy the requirements
for compatibility, in contradiction our assumption.
That is, $\sigma_{1,j}\neq \sigma_{2,j}=\sigma_{2,j^\prime}$
and $\sigma_{1,j^\prime}\neq \sigma_{2,j^\prime}=\sigma_{2,j}$.
This completes the proof of type (a).

The case of type (b) is symmetric.

For the case of type (c), the contrary assumption implies that,
for every pair of different $i$ and $j$,
one of the following holds:
$\sigma_{2,i}=\sigma_{3,j}$, $\sigma_{1,i}=\sigma_{3,j}$,
$\sigma_{2,j}=\sigma_{3,i}$, or $\sigma_{1,j}=\sigma_{3,i}$.
We increase $c_i$ for the two first cases
and increase $c_j$ for the last two cases.
Then $\sum_{i=1}^6c_i\geq \binom{6}{2}=15$.
Therefore, there is an $i$ with $c_i\geq 3$.
Hence, there are $j$ and $j^\prime$ with $j\neq j^\prime$
such that $\sigma_{3,j}=\sigma_{3,j^\prime}$,
implying that the two latter
are a compatible pair.
\end{proof}

\begin{proof}[Proof of~\Cref{t: pl one fma}]
Let
$\bsigma
\in L\left(\bA\right)$ be a word of length at least $11n^2+n$.
Take an accepting run of $\bA$ on $\bsigma$
and focus on the first $11n^2+n=(11n+1)n$ transitions.

If there are $n$ consecutive register-preserving
transitions, then there is also a repetition of states,
which implies that there is a repetition of configurations.
In this case, the pumping occurs by~\Cref{l: same cfg}.

Otherwise, every $n$ consecutive transitions contain
a register-updating transition.
Since we are considering $(11n+1)n$ transitions,
there are at least $11n+1$ register-updating transitions.

We now classify these updating transitions according
to the three types (a), (b), and (c).
If there are at least $3n+1$ transitions of type (a),
then four of them start in the same state,
and~\Cref{d: comp} together with~\Cref{p: comp tri}
yields a compatible pair.
The same argument applies to type (b),
using the target state instead of the source state.
Finally, if there are $5n+1$ transitions of type (c),
then six of them have the same target state,
and again~\Cref{p: comp tri} yields a compatible pair.
If none of these alternatives holds,
the total number of register-updating transitions
is at most $3n+3n+5n=11n$,
contradicting the previous paragraph.
Therefore, a compatible pair exists,
and the word is pumpable by~\Cref{l: diff registers}.
\end{proof}

\begin{remark}
\label{r: fma1}
For non-guessing one-register finite-memory automata,
this bound improves to $3n^2+n$,
because transitions of types (b) and (c) are impossible;
see~\Cref{r: tran guess}.
\end{remark}

\subsection{Minimal pumping length}

For automata over a finite alphabet,
it is known that the minimal constant needed for pumping
is at most the number of states.
In particular, the pumping length upper bound does not depend
on the size of the alphabet.
Thus, it is natural to ask whether
the pumping length constant in~\Cref{t: pl one fma}
can be improved to be linear in the number
of states of the automaton.

In this section, we show that
a quadratic blowup in the pumping constant is generally inevitable.
In particular,
the translation from quasi-regular language $L$
to its restriction $L\cap (\Sigma^\prime)^\ast$ to
a finite alphabet $\Sigma^\prime\subset\Sigma$,
is not necessarily linear in the number of states of
the automaton that accepts $L$ and
the size of the alphabet $\Sigma^\prime$.
It follows that a description of a regular language by
a finite-memory automaton is sometimes shorter than that
by an ordinary finite state automaton.

\begin{remark}
In fact,
it has already been conjectured that this translation might
be exponential, because the non-emptiness problem for
deterministic finite-memory automata
is known to be NP-complete~\cite[Theorem 4]{SakamotoI00}.
Note that the emptiness of a quasi-regular
language is equivalent
to the emptiness of its restriction to an alphabet of
size equal to the number of registers
(cf.~\cite[Theorem~1]{KaminskiF94}),
and the latter can be decided in linear time.
\end{remark}

\begin{example}
Let $\bA_{\text{sq,n}}$ be
a one-register finite-memory automaton,
with transitions depicted in~\Cref{fig: aut sq n}.
For readability, $=$ denotes the constraint $x_1=y=x_1^\prime$
and
$\neq_{\text{replace}}$ denotes $x_1\neq y=x_1^\prime$.
The language $L(\bA_{\text{sq,n}})$ consists of all the words
of the form,
\begin{equation}
\label{eq: example}
\prod_{i=1}^m (\tau_i)^{k_i}\#\tau_i (\tau_{i+1})^{k_i+1}\$\,,
\end{equation}
where
\begin{itemize}
\item $\#,\$ \in D\subset \Sigma$ are distinguished symbols;
\item $\tau_1,\tau_2,\ldots,\tau_{m+1}\in\Sigma\setminus\{\#,\$\}$
is a sequence of symbols such that
for all $i=1,2,\ldots,m$, $\tau_i\neq\tau_{i+1}$;
and
\item $k_1,k_2,\ldots,k_m\in\{1,2,\ldots,n\}$.
\end{itemize}
\end{example}

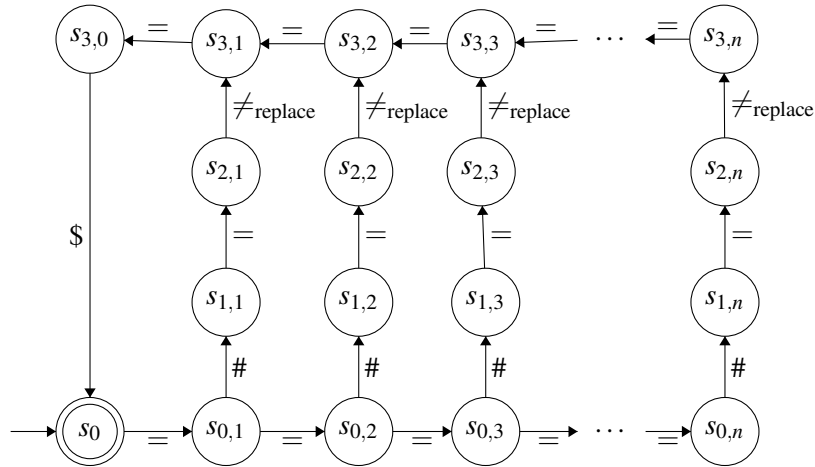
\begin{figure}[ht]
\begin{center}
\begin{tikzpicture}[scale=0.15]
\tikzstyle{every node}+=[inner sep=0pt]
\draw [black] (9.7,-44.1) circle (3);
\draw (9.7,-44.1) node {$s_0$};
\draw [black] (9.7,-44.1) circle (2.4);
\draw [black] (21.7,-44.1) circle (3);
\draw (21.7,-44.1) node {$s_{0,1}$};
\draw [black] (33.4,-44.1) circle (3);
\draw (33.4,-44.1) node {$s_{0,2}$};
\draw [black] (44.6,-44.1) circle (3);
\draw (44.6,-44.1) node {$s_{0,3}$};
\draw [black] (65.7,-44.1) circle (3);
\draw (65.7,-44.1) node {$s_{0,n}$};

\draw (55.7,-44.1) node {$\cdots$};

\draw [black] (21.7,-32.7) circle (3);
\draw (21.7,-32.7) node {$s_{1,1}$};
\draw [black] (33.4,-32.7) circle (3);
\draw (33.4,-32.7) node {$s_{1,2}$};
\draw [black] (44.6,-32.7) circle (3);
\draw (44.6,-32.7) node {$s_{1,3}$};
\draw [black] (65.7,-32.7) circle (3);
\draw (65.7,-32.7) node {$s_{1,n}$};

\draw [black] (21.7,-21.2) circle (3);
\draw (21.7,-21.2) node {$s_{2,1}$};
\draw [black] (33.4,-21.2) circle (3);
\draw (33.4,-21.2) node {$s_{2,2}$};
\draw [black] (44.2,-21.2) circle (3);
\draw (44.2,-21.2) node {$s_{2,3}$};
\draw [black] (65.7,-21.2) circle (3);
\draw (65.7,-21.2) node {$s_{2,n}$};

\draw [black] (21.7,-9.9) circle (3);
\draw (21.7,-9.9) node {$s_{3,1}$};
\draw [black] (33.4,-9.9) circle (3);
\draw (33.4,-9.9) node {$s_{3,2}$};
\draw [black] (44.2,-9.9) circle (3);
\draw (44.2,-9.9) node {$s_{3,3}$};
\draw [black] (65.6,-9.5) circle (3);
\draw (65.6,-9.5) node {$s_{3,n}$};
\draw [black] (9.7,-9.5) circle (3);
\draw (9.7,-9.5) node {$s_{3,0}$};
\draw (55.7,-9.5) node {$\cdots$};

\draw [black] (2.7,-44.1) -- (6.7,-44.1);
\fill [black] (6.7,-44.1) -- (5.9,-43.6) -- (5.9,-44.6);

\draw [black] (12.7,-44.1) -- (18.7,-44.1);
\fill [black] (18.7,-44.1) -- (17.9,-43.6) -- (17.9,-44.6);
\draw (15.7,-44.6) node [below] {$=$};
\draw [black] (24.7,-44.1) -- (30.4,-44.1);
\fill [black] (30.4,-44.1) -- (29.6,-43.6) -- (29.6,-44.6);
\draw (27.55,-44.6) node [below] {$=$};
\draw [black] (36.4,-44.1) -- (41.6,-44.1);
\fill [black] (41.6,-44.1) -- (40.8,-43.6) -- (40.8,-44.6);
\draw (39,-44.6) node [below] {$=$};
\draw [black] (47.6,-44.1) -- (52.7,-44.1);
\fill [black] (52.7,-44.1) -- (51.9,-43.6) -- (51.9,-44.6);
\draw (50.15,-44.6) node [below] {$=$};
\draw [black] (58.7,-44.1) -- (62.7,-44.1);
\fill [black] (62.7,-44.1) -- (61.9,-43.6) -- (61.9,-44.6);
\draw (60.7,-44.6) node [below] {$=$};
\draw [black] (21.7,-41.1) -- (21.7,-35.7);
\fill [black] (21.7,-35.7) -- (21.2,-36.5) -- (22.2,-36.5);
\draw (22.2,-38.4) node [right] {$\#$};
\draw [black] (33.4,-41.1) -- (33.4,-35.7);
\fill [black] (33.4,-35.7) -- (32.9,-36.5) -- (33.9,-36.5);
\draw (33.9,-38.4) node [right] {$\#$};
\draw [black] (44.6,-41.1) -- (44.6,-35.7);
\fill [black] (44.6,-35.7) -- (44.1,-36.5) -- (45.1,-36.5);
\draw (45.1,-38.4) node [right] {$\#$};
\draw [black] (65.7,-41.1) -- (65.7,-35.7);
\fill [black] (65.7,-35.7) -- (65.2,-36.5) -- (66.2,-36.5);
\draw (66.2,-38.4) node [right] {$\#$};
\draw [black] (21.7,-29.7) -- (21.7,-24.2);
\fill [black] (21.7,-24.2) -- (21.2,-25) -- (22.2,-25);
\draw (22.2,-26.95) node [right] {$=$};
\draw [black] (33.4,-29.7) -- (33.4,-24.2);
\fill [black] (33.4,-24.2) -- (32.9,-25) -- (33.9,-25);
\draw (33.9,-26.95) node [right] {$=$};
\draw [black] (44.5,-29.7) -- (44.3,-24.2);
\fill [black] (44.3,-24.2) -- (43.83,-25.02) -- (44.83,-24.98);
\draw (44.95,-26.94) node [right] {$=$};
\draw [black] (65.7,-29.7) -- (65.7,-24.2);
\fill [black] (65.7,-24.2) -- (65.2,-25) -- (66.2,-25);
\draw (66.2,-26.95) node [right] {$=$};
\draw [black] (21.7,-18.2) -- (21.7,-12.9);
\fill [black] (21.7,-12.9) -- (21.2,-13.7) -- (22.2,-13.7);
\draw (22.2,-15.55) node [right] {$\neq_{\text{replace}}$};
\draw [black] (33.4,-18.2) -- (33.4,-12.9);
\fill [black] (33.4,-12.9) -- (32.9,-13.7) -- (33.9,-13.7);
\draw (33.9,-15.55) node [right] {$\neq_{\text{replace}}$};
\draw [black] (44.2,-18.2) -- (44.2,-12.9);
\fill [black] (44.2,-12.9) -- (43.7,-13.7) -- (44.7,-13.7);
\draw (44.7,-15.55) node [right] {$\neq_{\text{replace}}$};
\draw [black] (65.67,-18.2) -- (65.63,-12.5);
\fill [black] (65.63,-12.5) -- (65.13,-13.3) -- (66.13,-13.3);
\draw (66.16,-15.35) node [right] {$\neq_{\text{replace}}$};
\draw [black] (62.6,-9.5) -- (58.7,-9.5);
\fill [black] (58.7,-9.5) -- (59.5,-10) -- (59.5,-9);
\draw (60.65,-9) node [above] {$=$};
\draw [black] (52.7,-9.6) -- (47.2,-9.8);
\fill [black] (47.2,-9.8) -- (48.02,-10.27) -- (47.98,-9.27);
\draw (49.92,-9.17) node [above] {$=$};
\draw [black] (41.2,-9.9) -- (36.4,-9.9);
\fill [black] (36.4,-9.9) -- (37.2,-10.4) -- (37.2,-9.4);
\draw (38.8,-9.4) node [above] {$=$};
\draw [black] (30.4,-9.9) -- (24.7,-9.9);
\fill [black] (24.7,-9.9) -- (25.5,-10.4) -- (25.5,-9.4);
\draw (27.55,-9.4) node [above] {$=$};
\draw [black] (18.7,-9.8) -- (12.7,-9.6);
\fill [black] (12.7,-9.6) -- (13.48,-10.13) -- (13.51,-9.13);
\draw (15.72,-9.17) node [above] {$=$};
\draw [black] (9.7,-12.5) -- (9.7,-41.1);
\fill [black] (9.7,-41.1) -- (10.2,-40.3) -- (9.2,-40.3);
\draw (9.2,-26.8) node [left] {$\$ $};
\end{tikzpicture}
\end{center}
\centering
\caption{Visual representation of $\bA_{\text{sq,n}}$. }
\label{fig: aut sq n}
\end{figure}

\begin{lemma}
\label{l: no pump fma}
Let
\[
\bsigma =
\prod_{i=1}^n (\tau_i)^i\# \tau_i(\tau_{i+1})^{i+1}\$ 
\, ,
\]
i.e., in~\eqref{eq: example}, $m = n$,
all $\tau_1,\tau_2,\ldots,\tau_{n+1}$ are pairwise different and
$k_i=i$ for $i=1,2,\ldots,n$.
Then~$\bsigma$ contains no nonempty pattern that can be pumped
with respect to $\langof{\bA_{\text{sq,n}}}$.
\end{lemma}

\begin{proof}
Assume, to the contrary, that
$\bsigma=\btau\bupsilon\bvarphi$, where
$\bupsilon\neq\varepsilon$, and that
\[
\btau\bupsilon^k\bvarphi
\in\langof{\bA_{\text{sq,n}}}
\qquad
\mbox{for all } k=0,1,2,\ldots.
\]

The factor $\bupsilon$ contains equally many occurrences of
$\#$ and $\$$.
It must also contain a separator, since otherwise pumping it
would create a separator-free factor of length greater than $n+2$.
Hence its separators have one of the forms
\[
(\#\$)^r
\qquad\mbox{or}\qquad
(\$\#)^r
\]
for some $r\geq 1$.

Suppose first that they have the form $(\#\$)^r$.
Let the first $\#$ belong to block $i$ and the last $\$$ to block $j$.
Then
\[
\bupsilon =
(\tau_i)^a\#\cdots\$(\tau_{j+1})^b
\]
for some $0\leq a\leq i$ and $0\leq b\leq j+1$.
At the junction of two copies of $\bupsilon$, the run
must read
\[
\$(\tau_{j+1})^b(\tau_i)^a\#\tau_i.
\]
After the $\$$, the register contains $\tau_{j+1}$.
Thus the run must read a nonempty power of $\tau_{j+1}$ before the
next $\#$ and must read $\tau_{j+1}$ immediately after it.
This is impossible because $\tau_i\neq\tau_{j+1}$.

Suppose now that the separators have the form $(\$\#)^r$.
Let the first $\$$ belong to block $i$ and the last $\#$ to block $j$,
where $i<j$.
Write
\[
\bupsilon
=\bchi\$\cdots\#\bchi^\prime,
\]
where $\bchi$ is a suffix of
$\tau_i(\tau_{i+1})^{i+1}$ and $\bchi^\prime$ is a prefix of
$\tau_j(\tau_{j+1})^{j+1}$.
At the junction of two copies, the run of the automaton forces that
\[
\bchi^\prime\bchi
=
\tau_j(\tau_{i+1})^{j+1}.
\]
If $j=i+1$, this contradicts the required change of register.
If $j>i+1$, then $\bchi^\prime=\tau_j$, because every longer prefix
contains $\tau_{j+1}$.
Hence
\[
\bchi
=
(\tau_{i+1})^{j+1},
\]
which is impossible because $\bchi$ is a suffix of
$\tau_i(\tau_{i+1})^{i+1}$ and $j+1>i+1$.

Thus, no nonempty factor of $\bsigma$ can be pumped.
\end{proof}

Note that $\bA_{\text{sq,n}}$ has $4n+2$ states
and $5n+1$ transitions. 
Thus, the description of this automaton is of size $O(n)$.
However, the word described in~\Cref{l: no pump fma}
(which cannot be pumped) is of length $n^2+5n$.
It follows that the minimal constant needed for pumping 
the language accepted by this automaton is 
at least quadratic in the automaton description size.

\begin{corollary}
\label{c: quad}
There is a family $(\bA_n)_{n\geq 1}$
of one-register finite-memory automata
of description size $O(n)$
such that the minimal classical pumping length of $\langof{\bA_n}$
is $\Omega(n^2)$.
\end{corollary}

\subsection{Computing minimal pumping lengths}
As we have seen in~\Cref{t: ud prob},
universality of the automaton language
reduces to computing the minimal pumping constants.
However, for one-register finite-memory automata with guessing,
universality is also undecidable~\cite[Exercise 9]{Bojanczyk19}
and, therefore,
so is computing their minimal pumping lengths.
If we disallow guessing, universality becomes decidable
but non primitive-recursive~\cite[Theorem 5.2]{DemriL09}.
In this section, we show that
it is possible to compute
the minimal classical pumping length
for non-guessing one-register automata.

\begin{theorem}
\label{t: fma1 pll}
For a one-register non-guessing finite-memory automaton $\bA$,
\[
p_{\text{min}}=\min\{p\in\mathbb{N}:L(\bA)\mbox{ is }
(p,1)\mbox{-pumpable}\}
\]
is computable.
\end{theorem}

The proof of~\Cref{t: fma1 pll}
immediately follows from the following lemma,
whose proof is postponed to the end of this section.

\begin{lemma}
\label{l: aux}
Let $\bA$ be a one-register non-guessing
finite-memory automaton and let  $p\in\mathbb{N}$.
Then
there is a computable constant $N_{A,p}$
such that
the following are equivalent:
\begin{enumerate}
\item $\langof{\bA}$ is $(p,1)$-pumpable.
\item Every word $\bsigma\in\langof{\bA}$ with
$p\leq |\bsigma|\leq p+N_{A,p}$
is $(p,1)$-pumpable with respect to $\langof{\bA}$.
\end{enumerate}
\end{lemma}

\begin{proof}[Proof of~\Cref{t: fma1 pll}]
By~\Cref{t: pl one fma} and~\Cref{r: fma1},
we have an upper bound $3n^2+n$ for $p_{\text{min}}$.
Then, we can search for the minimal pumping length
by checking if $\langof{\bA}$ is $(p,1)$-pumpable
for values $p$ in $[1,3n^2+n]$.

It remains to show that checking for specific value is decidable.
By~\Cref{l: aux} it suffices to verify that all words
of bounded length are pumpable.
Checking that a specific word $\bsigma$ is $(p,1)$-pumpable,
can be done by iterating all decompositions on the prefix
of length $p$ and deciding if
$\btau(\bupsilon)^\ast\bvarphi\subseteq \langof{\bA}\cap [\bsigma]^\ast$,
which is decidable, since both languages are in fact,
regular languages over the finite alphabet $[\bsigma]$.
However, over infinite alphabets,
there are infinitely many words, even of bounded lengths.
Nevertheless, it suffices to consider words that are equivalent
up to $D$-permutations, of which there are only finitely
many equivalence classes for each length.
\end{proof}

For a configuration $c$, let $L_c(\bA)$
be the language of all words accepted from configuration $c$.
For two finite sets of configurations, $C$ and $C^\prime$,
let
\begin{equation}
\label{eq: lcc}
L_{C,C^\prime}(\bA)=\bigcap_{c\in C}L_c(\bA)\cap
\bigcap_{c^\prime\in C^\prime}\overline{L_{c^\prime}(\bA)}\, ,
\end{equation}
that is all the words that are accepted from every configuration in $C$
and are not accepted from configurations in $C^\prime$.

For the proof of~\Cref{l: aux}, we use
the extended model of alternating non-guessing
one-register automata,
which is closed under Boolean operations and thus
can recognize languages of the form~\eqref{eq: lcc}.
It is known that for every one-register non-guessing alternating automaton,
one can compute a constant $N$ such that its language is nonempty
if and only if there is an accepted word of length
at most $N$~\cite[Lemmas 16 and 18]{GenkinKP14}.
Moreover, the constant depends only on the number of states and 
the number of constants in the description of the automaton.

\begin{proof}[Proof of~\Cref{l: aux}]
Let $N_{A,p}$ be the maximal constant for automata
that accept languages of the form~\eqref{eq: lcc}
with $|C|,|C^\prime|\leq n(p+1)$ where $n$
is the number of states
of $\bA$.

Assume that all words of length
at most $p+N_{A,p}$ are $(p,1)$-pumpable.
Let $\bsigma\in\langof{\bA}$ with $|\bsigma|>p+N_{A,p}$ and
let $\bsigma=\bomega\bpsi$ with $|\bomega|=p$.
Since $\bA$ is non-guessing,
every run of $\bA$ on $\bomega$ ends at a configuration
in $C_{\omega}=S\times ([\bomega]\cup\{\theta_0\})$.
Partition $C_\omega$ as
\[
C=\{c\in C_\omega:\bpsi\in L_c(\bA)\},
\qquad C^\prime=C_\omega\setminus C\, .
\]
Notice, $\bpsi\in L_{C,C^\prime}(\bA)$.
Moreover, since $\bsigma\in \langof{\bA}$,
there is an accepting run of $\bA$ on $\bsigma$,
let $c$ be the configuration after reading $\bomega$,
by definition $c\in C$.

Since $L_{C,C^\prime}(\bA)\neq\emptyset$,
there is a word $\bpsi^\prime\in L_{C,C^\prime}(\bA)$
with $|\bpsi^\prime|\leq N_{A,p}$
by the definition of $N_{A,p}$,
In particular, there is an accepting run of $\bA$ on
$\bomega\cdot \bpsi^\prime$
because $\bpsi^\prime\in L_c(\bA)$.
By our assumption, the word $\bomega\bpsi^\prime$ is $(p,1)$-pumpable.
Let the decomposition be $\btau\bupsilon\bvarphi\bpsi^\prime$
such that for all $k=0,1,2,\ldots$, 
$\btau\bupsilon^k\bvarphi\bpsi^\prime\in L(\bA)$.

We contend that this is also a decomposition
that works for $\bsigma$,
i.e., for all $k=0,1,2,\ldots$, 
$\btau\bupsilon^k\bvarphi\bpsi\in L(\bA)$.
Indeed, take an accepting run of $\bA$ on 
$\btau\bupsilon^k\bvarphi\bpsi^\prime\in L(\bA)$
it is of the form,
\[
c_0\transword{\btau\bupsilon^k\bvarphi}c_{\text{mid}}
\transword{\bpsi^\prime}c_f\, .
\]
Since $\bA$ is non-guessing, the register
at configuration $c_{\text{mid}}$
must be either the initial value $\theta_0$ or one of the symbols
seen in the run so far, thus
in $[\btau\bupsilon^k\bvarphi]\subseteq [\bomega]$.
Hence, $c_{\text{mid}}\in C_\omega$. 
Finally, since $\bpsi^\prime\in L_{c_{\text{mid}}}(\bA)$,
$\bpsi\in L_{c_{\text{mid}}}(\bA)$.
Therefore, there is an accepting run from
$c_{\text{mid}}$ over $\bpsi$ as well.
\end{proof}

\subsection{Extension to other models}
\label{s: hfma}

Lasota and Piorkowski~\cite{LasotaP21}
introduced a restriction of finite-memory automata
called \emph{hierarchical register automata}.
These automata 
update registers in a hierarchical manner: if
a transition rule changes the $i$-th register,
then all smaller registers must be preserved,
while larger registers are left unconstrained.
The formal definition is as follows.

\begin{definition}
A hierarchical $r$-register automaton 
is a finite-memory automaton where each transition
constraint $\varphi$ has the following form:
\[
\varphi=\psi(x_1,x_2,\ldots,x_i,y,x_i^\prime)
\wedge \bigwedge_{1\leq j<i}(x_j=x_j^\prime)\, ,
\]
for some $i\in\{1,2,\ldots,r\}$.
\end{definition}

Note that all smaller registers are preserved, and larger ones are
not mentioned in $\varphi$ 
(and hence their pre- and post-values are unspecified, which means that
any pre- and post-values are allowed).
The number $i$ is called
the level of the transition constraint.

Hierarchical register automata
recognize a proper subclass of all quasi-regular languages.
However, it is known that their languages
have rational commutative images~\cite[Theorem 14]{LasotaP21},
an extension of a result that was previously known
only for one-register automata~\cite[Theorem 6]{HofmanJLP21}.
In some sense, they are pseudo-one-register automata.
To strengthen this analogy, we show
that hierarchical register automata also
satisfy the classical pumping lemma.

\begin{theorem}
\label{t: pl hfma}
Let $\bA$ be a hierarchical finite-memory automaton with $n$ states and
$r$ registers.  Then $L(\bA)$ is $((11n+1)^rn,1)$-pumpable.
\end{theorem}

For the proof, we extend the notion of
compatibility of transitions for this model.

\begin{definition}
Two register-updating transitions of level $i$:
$\conf{s}{(\sigma_1,\sigma_2,\ldots,\sigma_r)}\trans{\sigma_{r+1}}
\conf{s^\prime}{(\sigma^\prime_1,\sigma^\prime_2,\ldots,\sigma^\prime_r)}$
and
$\conf{t}{(\tau_1,\tau_2,\ldots,\tau_r)}\trans{\tau_{r+1}}
\conf{t^\prime}{(\tau^\prime_1,\tau^\prime_2,\ldots,\tau^\prime_r)}$,
are called compatible,
if 
\begin{itemize}
\item $\sigma_j=\tau_j$ for all $j<i$; and
\item the restrictions
$\conf{s}{\sigma_i}\trans{\sigma_{r+1}}\conf{s^\prime}{\sigma^\prime_i}$
and
$\conf{t}{\tau_i}\trans{\tau_{r+1}}\conf{t^\prime}{\tau_i^\prime}$,
of the above transitions
to the $i$-th register,
are compatible by~\Cref{d: comp}.
\end{itemize}


\end{definition}

Similarly to the proof of~\Cref{l: diff registers},
one can show that if there are two
compatible transitions, then there is a straightforward pumping.
It remains to prove that every run either
has repeating configurations
or contains two compatible register-updating transitions.

\begin{proof}[Proof of~\Cref{t: pl hfma}]
Let $\bsigma\in\langof{\bA}$ with $|\bsigma|\geq (11n+1)^rn$,
take an accepting run of $\bA$ on $\bsigma$
and focus on the first $(11n+1)^rn$ transitions.

First, if there are $11n+1$ register-updating transitions of level $1$,
then there are two compatible transitions, because
there are no constraints on the other register values.

Assume there are at most $11n$ level $1$ register-updating
transitions.
This partition the run to $11n+1$ sub-runs
in which the first register does not change.
Note that
\[
\frac{(11n+1)^rn-11n}{11n+1}>(11n+1)^{r-1}n-1\, .
\]
Therefore, there is a sub-run having at
least $(11n+1)^{r-1}n$ transitions
in which the first register does not change.

Now, if there are ${11n+1}$ register-updating transitions
of level $2$ with register $1$ preserved between them,
then there is a pair of compatible transitions of level $2$.

If there are at most $11n$ level $2$ register-updating transitions,
then the current sub-run is partitioned into at most $11n+1$
sub-runs in which the first two registers are preserved.
One of these sub-runs has at least $(11n+1)^{r-2}n$ transitions.
Repeating this argument for levels $3,4,\ldots,r-1$,
we obtain a sub-run of length at least $(11n+1)n$
in which the first $r-1$ registers are preserved.

At the end, the final sub-run has at least $(11n+1)n=11n^2+n$ transitions
in which all the first ${r-1}$ registers are preserved.
Therefore,
either there are $n$ consecutive transitions in which all the registers
are preserved or there are $11n+1$ register-updating transitions
which are necessarily of level $r$.
In the former case, there is a repetition of configurations,
and in the latter case there is a compatible pair of transitions.
\end{proof}

\section{Conclusions}
We studied quantitative pumping lemmas
for finite-memory automata
over infinite alphabets.
For general finite-memory automata,
we improved the known bounds on the order
of the alphabet permutation
needed for pumping, and showed that order two always suffices
if the pumping length is sufficiently long.
For one-register automata,
and more generally for hierarchical register automata,
we proved that the classical pumping lemma holds without any
nontrivial permutation.
In the one-register case, the quadratic upper bound
on the pumping length is optimal
up to constant factors.

The undecidability results show that optimal
pumping constants cannot be computed in general.
A natural direction for further work is to identify
additional subclasses of finite-memory automata,
beyond the non-guessing one-register model,
for which optimal or near-optimal pumping constants
are decidable.
Another question is whether the order-two pumping length
for general finite-memory automata
can be reduced substantially from the current Ramsey-type bound.

\bibliographystyle{eptcs}
\bibliography{generic}

\end{document}